\documentclass[letterpaper, 10 pt, conference]{ieeeconf} 

\IEEEoverridecommandlockouts                              
\usepackage{booktabs} 
\usepackage{latexsym,amsfonts,amsmath}
\usepackage{mathrsfs}
\usepackage{graphicx}
\usepackage{wrapfig}
\usepackage{paralist}
\usepackage{dsfont}
\usepackage{eso-pic}
\usepackage{epsfig}
\usepackage{ragged2e}
\usepackage{mathtools}
\usepackage[nocomma]{optidef}
\usepackage{epstopdf}
\usepackage{lipsum}
\usepackage{varwidth}
\usepackage{fancyhdr}
\usepackage{multirow}

\newtheorem{theorem}{Theorem}[section]
\newtheorem{lemma}[theorem]{Lemma}

\newtheorem{problem}[theorem]{Problem}

\newtheorem{definition}[theorem]{Definition}

\newtheorem{remark}[theorem]{Remark}

\usepackage{graphicx}
\usepackage{subfig} 
\usepackage{caption}

\usepackage{tikz}
\usetikzlibrary{arrows,calc,decorations.markings,math,arrows.meta}
\usepackage{pgfplots}

\usepackage{algorithm}
\usepackage{algorithmic}

\usepackage{xcolor}         
\definecolor{rangyek}{RGB}{0, 75, 255}
\definecolor{rangdo}{RGB}{239, 62, 91}

\DeclareRobustCommand\sampleline[1]{%
	\tikz\draw[#1] (0,0) (0,\the\dimexpr\fontdimen22\textfont2\relax)
	-- (2em,\the\dimexpr\fontdimen22\textfont2\relax);%
}

\usepackage[utf8]{inputenc} 
\usepackage[T1]{fontenc}    
\usepackage[noadjust]{cite}
\makeatletter
\let\NAT@parse\undefined
\makeatother
\usepackage[colorlinks=true, linkcolor=blue,citecolor=blue,urlcolor=blue]{hyperref}

\usepackage{booktabs}       
\usepackage{amsfonts}       
\usepackage{nicefrac}       
\usepackage{microtype}      

\usepackage{stackengine}
\usepackage{scalerel}

\usepackage[many]{tcolorbox}
\usetikzlibrary{calc}
\tcbuselibrary{skins}

\makeatletter
\newsavebox\myboxA
\newsavebox\myboxB
\newlength\mylenA

\newcommand*\xoverline[2][0.75]{%
	\sbox{\myboxA}{$\m@th#2$}%
	\setbox\myboxB\null
	\ht\myboxB=\ht\myboxA%
	\dp\myboxB=\dp\myboxA%
	\wd\myboxB=#1\wd\myboxA
	\sbox\myboxB{$\m@th\overline{\copy\myboxB}$}
	\setlength\mylenA{\the\wd\myboxA}
	\addtolength\mylenA{-\the\wd\myboxB}%
	\ifdim\wd\myboxB<\wd\myboxA%
	\rlap{\hskip 0.5\mylenA\usebox\myboxB}{\usebox\myboxA}%
	\else
	\hskip -0.5\mylenA\rlap{\usebox\myboxA}{\hskip 0.5\mylenA\usebox\myboxB}%
	\fi}
\makeatother

\newtcolorbox{resp}[1][]{%
	enhanced jigsaw,%
	colback=gray!5!white,%
	colframe=gray!80!black,%
	size=small,%
	boxrule=1pt,%
	halign title=flush center,%
	coltitle=black,%
	breakable,%
	drop shadow=black!50!white,%
	attach boxed title to top left={xshift=1cm,yshift=-\tcboxedtitleheight/2,yshifttext=-\tcboxedtitleheight/2},%
	minipage boxed title=3cm,%
	boxed title style={%
		colback=white,%
		size=fbox,%
		boxrule=1pt,%
		boxsep=2pt,%
		underlay={%
			\coordinate (dotA) at ($(interior.west) + (-0.5pt,0)$);
			\coordinate (dotB) at ($(interior.east) + (0.5pt,0)$);
			\begin{scope}[gray!80!black]
				\fill (dotA) circle (2pt);
				\fill (dotB) circle (2pt);
			\end{scope}
		}%
	},%
	#1%
}

\newtcolorbox{mybox}[2][]{enhanced,
	attach boxed title to top left={xshift=1cm,yshift=-2mm},
	fonttitle=\bfseries,varwidth boxed title=0.7\linewidth,
	colbacktitle=gray!45!white,coltitle=gray!10!black,colframe=gray!50!black,
	interior style={top color=gray!5!white,bottom color=gray!0!white},
	boxed title style={boxrule=0.75mm,colframe=white,
		borderline={0.1mm}{0mm}{gray!50!black},
		borderline={0.1mm}{0.75mm}{gray!50!black},
		interior style={top color=gray!30!white,bottom color=gray!5!white,
			middle color=gray!50!white},
		drop fuzzy shadow},
	title={#2},#1}

\usepackage{xspace}
\newcommand{\R}{{\mathbb{R}}}
\newcommand{\N}{{\mathbb{N}}}
\newcommand{\ie}{{\it i.e.}}

\newcommand{\Vp}{{\mathcal{L}\mathbf V(x,\tilde x)}}
\newcommand{\V}{{\mathbf V(x,\tilde x)}}
\newcommand{\VecF}{{F}}

\DeclareFontFamily{U}{stix2bb}{}
\DeclareFontShape{U}{stix2bb}{m}{n} {<-> stix2-mathbb}{}

\NewDocumentCommand{\stixbbdigit}{m}{%
	\text{\usefont{U}{stix2bb}{m}{n}#1}%
}
\newcommand{\bbzero}{\stixbbdigit{0}}

\usepackage{xcolor}
\definecolor{lightblue}{rgb}{0.30, 0.75, 0.93}
\definecolor{lightblue}{rgb}{0.30, 0.75, 0.93}
\definecolor{mycolor}{rgb}{0, 0.45, 0.75}
\definecolor{mycolor1}{rgb}{0.39, 0.83, 0.07}
\definecolor{fluorescentpink}{rgb}{1.0, 0.08, 0.58}
\definecolor{royalblue(web)}{rgb}{0.25, 0.41, 0.88}
\definecolor{vividcerise}{rgb}{0.85, 0.11, 0.51}
\definecolor{tangelo}{rgb}{0.98, 0.3, 0.0}
\definecolor{persiangreen}{rgb}{0.0, 0.65, 0.58}
\definecolor{lemon}{rgb}{1.0, 0.97, 0.0}
\definecolor{gdash}{rgb}{0.10,0.82,0.10}

\definecolor{fluorescentpink}{rgb}{1.0, 0.08, 0.58}
\definecolor{D5Mr}{rgb}{0.62, 0.0, 1.0}
\definecolor{cUdg}{rgb}{0.3, 0.73, 0.09}
\definecolor{QKVd}{rgb}{0.03, 0.57, 0.82}
\definecolor{RE12}{rgb}{0.48, 0.25, 0.0}

\title{
Certified Learning of Incremental ISS Controllers for Unknown Nonlinear Polynomial Dynamics
}

 \author{Mahdieh Zaker, \IEEEmembership{Student Member,~IEEE}, David Angeli, \IEEEmembership{Fellow,~IEEE},\\ and Abolfazl Lavaei, \IEEEmembership{Senior Member,~IEEE}
 \thanks{M. Zaker and A. Lavaei are with the School of Computing, Newcastle University, United Kingdom. D. Angeli is with the Department of Electrical and Electronic Engineering, Imperial College London, United Kingdom, and the Department of Information Engineering, University of Florence, Italy.
 Emails: \{{\tt\small\href{mailto:mahdieh.zaker@newcastle.ac.uk}{mahdieh.zaker}, \href{mailto:abolfazl.lavaei@newcastle.ac.uk}{abolfazl.lavaei}\}@newcastle.ac.uk, \href{mailto:d.angeli@imperial.ac.uk}{d.angeli@imperial.ac.uk}}. }%
 }

\begin{document}

\maketitle
\thispagestyle{empty}
\pagestyle{empty}

\begin{abstract}
Incremental input-to-state stability ($\delta$-ISS) offers a robust framework to ensure that small input variations result in proportionally minor deviations in the state of a nonlinear system. This property is essential in practical applications where input precision cannot be guaranteed. However, analyzing $\delta$-ISS demands \emph{precise knowledge} of system dynamics to assess the state's incremental response to input changes, posing a challenge in real-world scenarios where mathematical models are unknown. In this work, we develop a data-driven approach to design $\delta$-ISS Lyapunov functions together with their corresponding $\delta$-ISS controllers for continuous-time input-affine nonlinear systems with polynomial dynamics, ensuring the $\delta$-ISS property is achieved without requiring knowledge of the system dynamics. In our data-driven scheme, we collect only two sets of input-state trajectories from sufficiently excited dynamics. By fulfilling a specific rank condition, we design $\delta$-ISS controllers using the collected samples through formulating a sum-of-squares optimization program. The effectiveness of our data-driven approach is evidenced by its application to a physical case study.
\end{abstract}

\section{Introduction}\label{sec:intro}
Most phenomena and engineering systems exhibit nonlinear behaviors with continuous-time evolutions, forming the cornerstone of numerous real-world applications, ranging from aerospace systems and chemical process control to renewable energy grids and healthcare technologies. The inherent complexity of these systems, stemming from their nonlinear dynamics, often makes traditional linearization-based methods inadequate for ensuring robust performance and stability under realistic operating conditions. Hence, their stability analysis and controller synthesis have long been recognized as pivotal challenges within the field of control theory.

A stronger property than stability for nonlinear systems is \emph{incremental} input-to-state stability ($\delta$-ISS), which is a more robust property that compares arbitrary trajectories with one another rather than with an equilibrium point. Just as Lyapunov functions are fundamental in stability analysis, $\delta$-ISS Lyapunov functions play a crucial role in examining incremental stability~\cite{angeli2002lyapunov}. However, not all nonlinear systems inherently exhibit the $\delta$-ISS property, necessitating the synthesis of a controller to enforce it, which is particularly valuable for analyzing complex nonlinear systems in real-world applications. Examples of such applications encompass synchronizing interconnected oscillators~\cite{stan2007analysis}, constructing symbolic models~\cite{pola2008approximately}, modeling nonlinear analog circuits~\cite{bond2010compact}, and synchronizing cyclic feedback systems~\cite{hamadeh2011global}.

Analyzing $\delta$-ISS requires knowledge of system dynamics to accurately evaluate the state's incremental response to variations in inputs. However, in many real-world applications, the exact mathematical models of nonlinear systems are either unavailable or difficult to obtain. To overcome this challenge, data-driven methods have been developed in two main categories: \emph{indirect} and \emph{direct} methods~\cite{dorfler2022bridging}. Indirect approaches focus on approximating unknown system dynamics using identification techniques; however, accurately deriving a mathematical model is often computationally intensive, especially for complex nonlinear systems~\cite{Hou2013model}. Even when system identification succeeds, designing a controller that ensures $\delta$-ISS for the identified model remains a significant hurdle. Thus, the complexity lies in two stages: (i) model identification and (ii) controller synthesis using traditional model-based methods. On the other hand, \emph{direct} data-driven approaches bypass system identification altogether, enabling stability analysis of the system directly from observed data~\cite{dorfler2022bridging}.

\textbf{Related Literature.} There have been several studies on the incremental stability analysis for nonlinear systems. Inspired by the concept of incremental ISS introduced by~\cite{angeli2002lyapunov} for continuous-time nonlinear systems, \cite{tran2016incremental} extends the study of this property to \emph{discrete-time} systems. Several new notions and conditions for incremental stability of \emph{hybrid} systems have been introduced  based on graphical closeness of solutions~\cite{biemond2018incremental}, a specific class of discontinuous dynamical systems containing a hybrid integrator using a small-gain approach~\cite{van2023small},  and a class of recurrent neural networks~\cite{d2023incremental}.

Additionally, incremental stability analysis has been conducted for Lurie systems in~\cite{su2024incremental}, which involve the feedback interconnection of a linear time-invariant system and a slope-bounded nonlinearity. The $\delta$-ISS of \emph{stochastic} switched systems has also been explored in~\cite{ren2021stochastic}. Furthermore, various studies have focused on synthesizing controllers to enforce incremental stability, such as developing control laws for systems in feedforward form~\cite{giaccagli2024incremental} and employing backstepping techniques for nonlinear (parametric-)strict-feedback systems~\cite{zamani2011backstepping}, a class of non-smooth control systems~\cite{zamani2013controller}, and stochastic Hamiltonian systems with jumps~\cite{jagtap2017backstepping}  (see \cite{kokotovic2001constructive} for more details about ISS property).

While system dynamics are essential in all these studies, some works also explore stability analysis and controller synthesis, though not incremental, through data-driven methods. In this regard, the work \cite{guo2021data} introduces a data-driven framework for stability analysis of persistently excited nonlinear polynomial systems. The study in \cite{berberich2020data} presents a robust data-driven model-predictive control approach for linear time-invariant systems, while \cite{taylor2021towards} develops a method for robust control synthesis of nonlinear systems under model uncertainty. Additionally, \cite{harrison2018control} proposes a data-driven strategy for learning controllers that can rapidly adapt to unknown dynamics. The work \cite{vzikelic2024compositional} proposes a methodology for composing neural network policies in stochastic settings, offering a formal probabilistic certification to ensure the satisfaction of specific policy behavior requirements.

Data-driven stability analysis has also been explored for switched systems in \cite{kenanian2019data}, continuous-time systems in \cite{boffi2021learning}, and discrete-time systems in \cite{lavaei2022data}. The work \cite{zhou2022neural} introduces a framework for stabilizing unknown nonlinear systems by jointly learning a neural Lyapunov function and a nonlinear controller, ensuring stability guarantees using Satisfiability Modulo Theories (SMT) solvers. Lastly, \cite{chen2024data} employs overapproximation techniques to define the set of polynomial dynamics consistent with noisy measured data, enabling the construction of an input-to-state stable Lyapunov function and a corresponding ISS control law for unknown nonlinear input-affine systems with polynomial dynamics.

Despite the substantial contributions of these studies, none addresses the analysis of the $\delta$-ISS property or the synthesis of $\delta$-ISS controllers for systems with unknown dynamics. In recent work~\cite{sundarsingh2024backstepping}, a Gaussian process as an \emph{indirect} data-driven approach is employed to learn the unknown dynamics of a specific class of control systems, followed by the design of a backstepping control scheme to ensure incremental input-to-state practical stability, a relaxed form of $\delta$-ISS. In contrast, through developing a \emph{direct} data-driven approach, our work circumvents the need to learn the system dynamics and derive a control law that renders the system $\delta$-ISS, relying solely on \emph{two input-state trajectories} of the system. Moreover, \cite{zaker2025data} proposes a \emph{direct} data-driven approach to analyze the incremental global asymptotic stability of nonlinear homogeneous networks by verifying the subsystems' $\delta$-ISS property, whereas our method focuses on \emph{controller synthesis}, a considerably more challenging task.

\textbf{Core Contributions.} We propose an innovative data-driven methodology for the design of $\delta$-ISS Lyapunov functions together with $\delta$-ISS controllers tailored to continuous-time nonlinear polynomial systems. Our approach aims to guarantee the $\delta$-ISS property while entirely bypassing the need for explicit knowledge of the underlying system dynamics. Unlike traditional model-based methods, our framework relies solely on data collected from the system, thereby enhancing its applicability in scenarios where obtaining accurate models is impractical or costly. Specifically, the proposed scheme leverages two sets of input-state trajectory data obtained from sufficiently excited system dynamics. By satisfying a specific rank condition on the collected data, our approach enables the synthesis of $\delta$-ISS controllers through the formulation of a sum-of-squares (SOS) optimization. We present a physical case study with unknown dynamics to validate our proposed methodology, highlighting its practicality and effectiveness. The results illustrate the robustness and versatility of our approach with respect to external inputs in achieving incremental stability for nonlinear systems, even in the absence of explicit system models. 

{\bf Organization.} The rest of the paper is structured as follows. Section~\ref{sec:2} introduces the mathematical preliminaries, notation, formal definitions of the system dynamics and the $\delta$-ISS property, along with the problem formulation. Section~\ref{sec:3} outlines the data collection procedure and presents our main theoretical results, which guarantee the $\delta$-ISS property for systems with unknown polynomial dynamics. Section~\ref{sec:4} provides simulation results to demonstrate the effectiveness of our proposed approach, while Section~\ref{sec:5} concludes the paper.

\section{Problem Description}\label{sec:2}
\textbf{Notation.} We denote the sets of all real numbers,  non-negative and positive real numbers by $\mathbb{R}$, $\mathbb{R}_0^+$, and $\mathbb{R}^+$, respectively. Moreover, the sets of non-negative and positive integers are represented as $\mathbb{N} := \{0, 1, 2, \ldots\}$ and $\mathbb{N}^+ := \{1, 2, \ldots\}$, respectively. The Euclidean norm of a vector $x \in \mathbb{R}^n$ is expressed as $\vert x \vert$, whereas the induced 2-norm of a matrix $A\in\R^{n\times m}$ is denoted by $\Vert A \Vert$. Given a signal $u(\cdot)$, the supremum of $u$ is denoted by $\vert u\vert_{\infty} := \text{sup}\{\vert u(t)\vert,t\geq 0\}$. For any square matrix $P$, the minimum and maximum eigenvalues are denoted by $\lambda_{\min}(P)$ and $\lambda_{\max}(P)$, respectively. The notation $P \succ 0$ $(P\succeq 0)$ indicates that a \emph{symmetric} matrix $P \in \mathbb{R}^{n \times n}$ is positive (semi-)definite, implying that all its eigenvalues are positive (non-negative). The transpose of a matrix $P$ is represented by $P^\top$. An $n \times n$ identity matrix is denoted by $\mathds I_{n}$, while $\bbzero_{n}$ represents an $n$-dimension vector with zero elements. The horizontal concatenation of vectors $x_i \in \R^n$ into an $n \times N$ matrix is written as $\begin{bmatrix} x_1 & \hspace{-0.2cm} x_2 & \hspace{-0.2cm} \dots & \hspace{-0.2cm} x_N \end{bmatrix}$. 
Functions are classified into specific types based on their properties: A function $\beta\!: \mathbb{R}_0^+ \rightarrow \mathbb{R}_0^+$ is defined as a $\mathcal{K}$ function if it is continuous, strictly increasing, and satisfies $\beta(0) = 0$. It is further categorized as a $\mathcal{K}_\infty$ function if it asymptotically approaches infinity as its argument goes to infinity. Moreover, a function $\beta\!: \mathbb{R}_0^+ \times \mathbb{R}_0^+\rightarrow \mathbb{R}_0^+$ belongs to the class $\mathcal{KL}$ if, for each fixed $s$, $\beta(r, s)$ is a $\mathcal{K}$ function with respect to $r$ and, for each fixed $r > 0$, $\beta(r, s)$ decreases with respect to $s$ and converges to zero as $s$ approaches infinity.

Motivated by the vast application of continuous-time nonlinear systems, as discussed in the introduction, we conduct our analysis on continuous-time input-affine nonlinear systems with polynomial dynamics, as formalized in the following definition.

\begin{definition}\label{def:ct-NPS}
	A continuous-time input-affine nonlinear system with polynomial dynamics (\textsc{ctia-NSP}) is described by 
	\begin{align}\label{eq:sys}
		\Omega\!: \dot x(t)=A\VecF(x(t)) + Bu(t),
	\end{align}
	where $A \in \R^{n\times N}$ is the system matrix, $B \in \R^{n\times m}$ is the control input matrix, while $u\in \R^{m}$ represents the control input. Here, $\VecF(x(t)) \in \R^{N}$ is a monomial vector in states $x\in \R^{n}$ such that $\VecF(\bbzero_{n}) = \bbzero_{N}$. We employ the tuple $\Omega\!=\!(A,B,\R^n,\R^N ,\R^m)$ to denote the system in~\eqref{eq:sys}. 
\end{definition}

We consider both matrices $A$ and $B$ to be \emph{unknown} and the exact knowledge of $\VecF(x)$ to be unavailable, yielding the system in~\eqref{eq:sys} unknown. However, we are provided with the maximum degree of $\VecF(x)$, enabling us to consider all possible combinations of monomials up to that degree or an exaggerated dictionary (\emph{i.e.,} family of functions) that contains the actual monomials of the system and possibly some additional irrelevant terms. With a slight abuse of notation, we use $\VecF(x)$ to represent both the actual and the exaggerated dictionary throughout the paper.

We now formally define the incremental input-to-state stability property in the upcoming definition~\cite{angeli2002lyapunov} for the \textsc{ctia-NSP} in~\eqref{eq:sys}.

\begin{definition}\label{def:delta-ISS}
	A \textsc{ctia-NSP} $\Omega\!=\!(A,B,\R^n,\R^N, \R^m)$ is incremental input-to-state stable ($\delta$-ISS) if there exist functions $\beta\in\mathcal{KL}$ and $\gamma\in\mathcal K_\infty$ such that for any arbitrary $x(0), \tilde x(0)\in\R^n$ and any arbitrary
	pair of locally essentially bounded, measurable input signals $u, \tilde u\in \mathscr U \subset \R^m$, with $\mathscr U$ being a closed convex set, one has
	\begin{align}
		\vert x(t)-\tilde x(t)\vert \leq \beta(\vert x(0)-\tilde x(0)\vert, t)+\gamma(\vert u-\tilde u\vert_{\infty})
	\end{align}
	 for $t>0$, where $\tilde x(t)$ is the state trajectory under the initial condition $\tilde x(0)\neq x(0)$ and the input signal $\tilde u$. The system is locally $\delta$-ISS if the property holds for sufficiently close initial states and input signals, and is incrementally globally asymptotically stable ($\delta$-GAS) for any arbitrary initial states when $u=\tilde u$.
\end{definition}

We hereby present the subsequent theorem, borrowed from~\cite{angeli2002lyapunov}, which elucidates the sufficient conditions for a \textsc{ctia-NSP} to achieve $\delta$-ISS, as specified in Definition~\ref{def:delta-ISS}.

\begin{theorem}
	Given a \textsc{ctia-NSP} $\Omega\!=\!(A,B,\R^n,\R^N,$ $ \R^m)$, assume there exist a smooth function $\mathbf V\!:\R^n\times\R^n\to\R_0^+$, referred to as a $\delta$-ISS Lyapunov function, and constants $ \underline{\alpha},\overline{\alpha},\epsilon \in \R^+,\rho\in \R^+_0$, satisfying
	\begin{itemize}
		\item[$\bullet$] $\forall x, \tilde x\!\in\! \R^n\!\!:$
		\begin{subequations}
			\begin{equation}\label{eq:dISS-con1-dis}
				\underline{\alpha}\vert x - \tilde x \vert^2\le \V\le \overline{\alpha}\vert x - \tilde x\vert^2,
			\end{equation}
			\item[$\bullet$] $\forall x, \tilde x\in\! \R^n, \: \forall u, \tilde u\in \R^m\!\!:$
			\begin{align}\label{eq:ISS-con2-dis}
				\Vp \leq -\epsilon \V + \rho \vert u - \tilde u\vert^2,
			\end{align}
		\end{subequations}
	\end{itemize}
	with
	\begin{align}
		\mathcal{L} \mathbf V(x, \tilde x) = &~\partial_x\V(A\VecF(x) + B u) \notag\\
		&+ \partial_{\tilde x}\V(A\VecF(\tilde x) + B \tilde u)\label{eq:Lie def}
	\end{align}
	representing the Lie derivative of $\mathbf V\!\!:\R^n\times\R^n\to\R_0^+$ with respect to the dynamics in~\eqref{eq:sys}, while $\partial_x \V = \frac{\partial \V}{\partial x}$ and $\partial_{\tilde x}\V=\frac{\partial\V}{\partial\tilde x}$. Then, the \textsc{ctia-NSP} $\Omega$ is  $\delta$-ISS in the sense of Definition~\ref{def:delta-ISS}.
\end{theorem}

While $\delta$-ISS Lyapunov functions facilitate the incremental stability analysis of \textsc{ctia-NSP}, constructing such functions is hindered by the unknown matrices $A$ and $B$, as appeared in \eqref{eq:Lie def}. Given this significant challenge, we now formally define the primary problem under investigation in this work.

\begin{resp}
	\begin{problem}\label{Problem-dis}
		Consider a \textsc{ctia-NSP} $\Omega=(A,B,\R^n,\R^N,$ $ \R^m)$, characterized by unknown matrices $A, B,$ and $\VecF(x)$, where only the maximum degree of $\VecF(x)$ or an exaggerated dictionary is  known. Design a $\delta$-ISS Lyapunov function $\mathbf V$ along with a corresponding $\delta$-ISS controller $u$ by leveraging solely two input-state trajectories collected from $\Omega$.
	\end{problem}
\end{resp}

To tackle Problem \ref{Problem-dis}, we present our data-driven approach in the subsequent section.
\section{Data-Driven Framework}\label{sec:3}
We choose our $\delta$-ISS Lyapunov function in a quadratic form as $\V = (x - \tilde x)^\top P(x - \tilde x)$, with $P \succ 0$. We collect $T \in \mathbb{N}^+$ samples with sampling time $\tau\in\R^+$ over the period $[t_0, t_0+(T - 1)\tau]$ from the unknown \textsc{ctia-NSP} in~\eqref{eq:sys} as follows:
\begin{align}\label{eq:data}
	\begin{array}{llllll}
		\mathbf U_{0,T} & \hspace{-0.2cm}= & \hspace{-0.2cm}[u(t_0) & \hspace{-0.2cm}u(t_0 + \tau) & \hspace{-0.2cm}\dots & \hspace{-0.2cm}u(t_0 +(T-1)\tau)],\\
		\mathbf X_{0, T} & \hspace{-0.2cm}= & \hspace{-0.2cm}[x(t_0) & \hspace{-0.2cm}x(t_0 + \tau) & \hspace{-0.2cm}\dots & \hspace{-0.2cm}x(t_0 + (T-1)\tau)],\\
		\mathbf X_{1, T} & \hspace{-0.2cm}= & \hspace{-0.2cm}[\dot x(t_0) & \hspace{-0.2cm}\dot x(t_0 + \tau) & \hspace{-0.2cm}\dots & \hspace{-0.2cm}\dot x(t_0 + (T-1)\tau)].
	\end{array}
\end{align}
The sampled data in~\eqref{eq:data} are referred to as a \emph{set of input-state trajectories}. Given the incremental nature of the property, we collect an additional set of input-state trajectories $\tilde{\mathbf X}_{0,T}$ and $\tilde{\mathbf X}_{1,T}$ based on $\tilde x$ with $\tilde{\mathbf U}_{0,T}$. Since $\mathbf U_{0,T}$ is \textit{arbitrary} input data during sample collection, we use the same input while collecting the other set of trajectories $\tilde{\mathbf X}_{0,T}$ and $\tilde{\mathbf X}_{1,T}$ from different initial conditions, \ie, $\tilde{\mathbf U}_{0,T} = \mathbf U_{0,T}$, whereas $\tilde x(t_0)\neq x(t_0)$ and $\dot{\tilde x}(t_0)\neq\dot x(t_0)$.
\begin{remark}
	Since the state derivatives represented in $\mathbf X_{1,T}$ cannot be directly measured, one can adopt two practical solutions to address this: (i) numerically approximating their values using $\dot x(t_0 + k\tau) = \frac{x(t_0+(k+1)\tau) - x(t_0 + k\tau)}{\tau}$, for $k\in\{0, \dots, T-1\}$, or (ii) estimating them through appropriate filters based on existing methods~\cite{larsson2008estimation, padoan2015towards}. Considering that both methods may introduce approximation errors (regarded as measurement noise), an approach akin to that in~\cite{guo2021data} can be utilized to account for the influence of errors in derivative approximations.
	In this approach, the data $\mathbf X_{1,T}$ is modeled as $\mathbf X_{1,T}=\hat{\mathbf X}_{1,T}+ \Delta$, where $\hat{\mathbf X}_{1,T}$ denotes the noise-free data and $\Delta$ captures the error as a small perturbation. The sole assumption is that $\Delta$ satisfies $\Delta\Delta^\top \preceq \Lambda\Lambda^\top$ for some known matrix $\Lambda$, a standard assumption in control that intuitively reflects bounded noise energy during data collection. However, as this work focuses on synthesizing controllers to ensure the $\delta$-ISS property, a fundamentally complex yet critical concept in control system analysis, we do not address the impact of measurement noise in the current analysis to maintain clarity of presentation. Nonetheless, it remains an important direction for future research and is currently under investigation.
\end{remark}

As the system dynamics, particularly the unknown matrices $A$ and $B$, appear in \eqref{eq:ISS-con2-dis} via \eqref{eq:Lie def}, we need to utilize the collected data in~\eqref{eq:data} to eliminate their dependency. To do so, we propose the next lemma to obtain the data-driven closed-loop representation of \textsc{ctia-NSP}, inspired by~\cite{guo2021data}.

\begin{lemma}\label{Lemma1-dis}
	Given a \textsc{ctia-NSP} $\Omega\!=\!(A,B,\R^n,\R^N, \R^m)$ and full \emph{row-rank} $(N\times T)$ matrices $\mathbf J_{0,T}$ and $\tilde{\mathbf J}_{0,T}$ defined as
	\begin{align}\label{eq:monomial traj-dis}
		\begin{array}{lllll}
			\mathbf J_{0,T} & \hspace{-0.2cm}= & \hspace{-0.2cm}\big[\VecF(x(t_0)) & \hspace{-0.2cm}\dots & \VecF(x(t_0 + (T-1)\tau))\big],\\
			\tilde{\mathbf J}_{0,T} & \hspace{-0.2cm}= & \hspace{-0.2cm}\big[\VecF(\tilde x(t_0)) & \hspace{-0.2cm}\dots & \VecF(\tilde x(t_0 + (T-1)\tau))\big],
		\end{array}
	\end{align}
	consider polynomial matrices $G(x)\in\R^{T\times n}$ and $G(\tilde x)\in\R^{T\times n}$ such that 
	\begin{subequations}\label{eq:lemma cons}
		\begin{align}\label{eq:lemma con1-dis}
			\aleph(x)& = 
			\mathbf J_{0,T}G(x),\\
			\aleph(\tilde x) &= 
			\tilde{\mathbf J}_{0,T}G(\tilde x),\label{eq:lemma con2-dis}
		\end{align}
	\end{subequations}
	where $\aleph(\cdot)$ is a transformation matrix satisfying		
		\begin{align}\label{new}	
			\VecF(x) &= \aleph(x)x,
		\end{align}
    for any $x\in\R^n$. Then, the closed-loop systems $A\VecF(x) + Bu$ and $A\VecF(\tilde x) + B\tilde u$, with the controllers designed as $u = \mathbf K(x)x + \hat u= \mathbf U_{0,T}G(x)x + \hat u$ and $\tilde u = \mathbf K(\tilde x)\tilde x + \hat{\tilde u}= \mathbf U_{0,T}G(\tilde x)\tilde x + \hat{\tilde u}$, where $\hat u$ and $\hat{\tilde u}$ are locally essentially bounded, measurable \textit{external} inputs, can be represented by their data-driven alternatives as
	\begin{equation}\label{eq:data-based cl}
		\begin{aligned}
			A \VecF(x)+ B u =\,& \mathbf X_{1,T} G(x)x+ B \hat u,\\
			A \VecF(\tilde x)+ B \tilde u =\,& \tilde{\mathbf X}_{1,T} G(\tilde x)\tilde x+ B \hat{\tilde u}.
		\end{aligned}
	\end{equation}
\end{lemma}\vspace{0.2cm}

\begin{proof}
	For the system dynamics \eqref{eq:sys}, by employing the trajectories $\mathbf X_{1,T}$ and $\mathbf U_{0,T}$ in~\eqref{eq:data} and matrix $\mathbf J_{0,T}$ in~\eqref{eq:monomial traj-dis}, we have
	\begin{align} \label{eq:data cl}
		\mathbf X_{1,T} &= A \mathbf J_{0,T} + B\mathbf U_{0,T} = [B\quad A]\begin{bmatrix}
			\mathbf U_{0,T}\\
			\mathbf J_{0,T}
		\end{bmatrix}\!\!.
	\end{align}
	By applying the controller $u \!=\! \mathbf K(x)x + \hat u \!=\! \mathbf U_{0,T}G(x)x + \hat u$, one has
	\begin{align}\notag
		A\VecF(x) + Bu   \overset{\eqref{new}}{=}&~(A\aleph(x) + B \mathbf K(x))x + B \hat u \\
		= &~[B\quad A] \begin{bmatrix}
			\mathbf K(x)\\
			\aleph(x)
		\end{bmatrix}x + B \hat u\notag\\
		\overset{\eqref{eq:lemma con1-dis}}{=} &~[B\quad A] \begin{bmatrix}
			\mathbf U_{0,T}\\
			\mathbf J_{0,T}
		\end{bmatrix}G(x)x+ B \hat u.\label{new3}
	\end{align}
	Now by applying~\eqref{eq:data cl} to~\eqref{new3}, we have 
	\begin{align*}
		A \VecF(x)+ B u &= \mathbf X_{1,T} G(x)x + B \hat u.
	\end{align*}
	Analogously, for the other trajectory $\tilde x$, by utilizing $\tilde{\mathbf X}_{1,T}$ and $\tilde{\mathbf U}_{0,T}=\mathbf U_{0,T}$ in \eqref{eq:data} and matrix $\tilde{\mathbf J}_{0,T}$ in~\eqref{eq:monomial traj-dis}, we have
	\begin{align*}
		\tilde{\mathbf X}_{1,T} =& A \tilde{\mathbf J}_{0,T} + B \mathbf U_{0,T}= [B\quad A]\begin{bmatrix}
			\mathbf U_{0,T}\\
			\tilde{\mathbf J}_{0,T}
		\end{bmatrix}\!.
	\end{align*}
	Then, by applying the controller designed as $\tilde u = \mathbf K(\tilde x)\tilde x + \hat{\tilde u}= \mathbf U_{0,T}G(\tilde x)\tilde x + \hat{\tilde u}$ and
	based on conditions~\eqref{eq:lemma con2-dis} and \eqref{new}, following the similar steps for $x$, one can attain
	\begin{align*}
		A \VecF(\tilde x)+ B \tilde u =\tilde{\mathbf X}_{1,T} G(\tilde x)\tilde x+ B \hat{\tilde u},
	\end{align*}
	which concludes the proof.
\end{proof}

It is worth noting that, with the controller in the form $u = \mathbf K(x)x + \hat u = \mathbf U_{0,T}G(x)x + \hat u$, the aim is to design $\mathbf K(x)$ such that the \textsc{ctia-NSP} is $\delta$-ISS with respect to the external input $\hat u$ (\emph{i.e.,} ensuring robustness against $\hat u$). A similar argument holds when $\tilde u$ is used, with $\hat{\tilde u}$ treated as the external input.

\begin{remark}
For the feasibility of the conditions outlined in~\eqref{eq:lemma cons}, it is required that the matrices $\mathbf J_{0,T}$ and $\tilde{\mathbf J}_{0,T}$ possess full row-rank. This requirement is essential for fulfilling constraints~\eqref{eq:lemma con1-dis} and \eqref{eq:lemma con2-dis}, when designing $G(x)$ and $G(\tilde x)$. The full row-rank of $\mathbf J_{0,T}$ and $\tilde{\mathbf J}_{0,T}$ effectively acts as a measure of data richness, highlighting the importance of employing adequately informative data. Since $\mathbf J_{0,T}$ and $\tilde{\mathbf J}_{0,T}$  are derived from collected data, their full row-rank assumption can be readily verified. 
\end{remark}

\begin{remark}
	Since $\VecF(\bbzero_{n}) = \bbzero_{N}$, it follows that transformation matrix $\aleph(x)$ can always be constructed to meet condition~\eqref{new}, without loss of generality. This enables our framework to focus on representations only involving $x$ and $\tilde{x}$, rather than $F(x)$ and $F(\tilde{x})$, in alignment with the structure of the $\delta$-ISS Lyapunov function $\V = (x - \tilde x)^\top P(x - \tilde x)$.
\end{remark}

Building upon the data-driven representation of the closed-loop system~\eqref{eq:data-based cl} provided in Lemma~\ref{Lemma1-dis}, we offer the following theorem, enabling the construction of both the $\delta$-ISS Lyapunov function and its associated $\delta$-ISS controller directly from data for the unknown \textsc{ctia-NSP} $\Omega$.

\begin{theorem}\label{Thm:main3-dis}
	Consider a \textsc{ctia-NSP} $\Omega=(A,B,\R^n,$ $\R^N, \R^m)$, characterized by unknown matrices $A$ and $B$. Given the closed-loop data-based representations in Lemma~\ref{Lemma1-dis}, if there exist polynomial matrices $\mathds Y(x)\in \R^{T\times n},$ and $\mathds Y(\tilde x)\in \R^{T\times n}$, and constant matrices $\Sigma\in\R^{n\times n}$, and $P\in\R^{n\times n}$, where $P\succ 0$, such that 
	\begin{subequations}\label{eq:main theorem conditions}
		\begin{align}
			\mathbf J_{0,T}\mathds Y(x) &= \aleph(x)P^{-1}\!,\label{eq:con1 SOS-dis}\\
			\tilde{\mathbf J}_{0,T}\mathds Y(\tilde x) &= \aleph(\tilde x)P^{-1}\!,\label{eq:con1 SOS'-dis}\\
			\mathbf X_{1,T}\mathds Y(x) &= \Sigma,\label{eq:Y on cl-dis}\\
			\tilde{\mathbf X}_{1,T}\mathds Y(\tilde x) &= \Sigma,\label{eq:Y' on cl-dis}\\
			\Sigma + \Sigma^\top + \vartheta\mathds I_n&\preceq -\epsilon P^{-1}\label{eq:con2 SOS-dis},
		\end{align}
		for some $\epsilon, \vartheta\in \R^+$, then $\V = (x - \tilde x)^\top P (x - \tilde x)$ is a $\delta$-ISS Lyapunov function with $\underline{\alpha} = \lambda_{\min}(P)$, $\overline{\alpha} = \lambda_{\max}(P)$, $ \rho =  \frac{\Vert B \Vert^2}{\vartheta}$, and $u =\mathbf K(x)x + \hat u$ is its $\delta$-ISS controller for $\Omega$ with
		\begin{align}\label{eq:K}
			\mathbf K(x) = \mathbf U_{0,T}\mathds Y(x)P.
		\end{align}
	Moreover, the \textsc{ctia-NSP} $\Omega$ is $\delta$-GAS if $\hat u = \hat{\tilde u}$.
	\end{subequations}
\end{theorem}
\begin{proof}
	Firstly, we show that condition~\eqref{eq:dISS-con1-dis} holds for the candidate Lyapunov function. Since
	\begin{align*}
		\lambda_{\min}(P)\vert x - \tilde x\vert^2 \leq \underbrace{(x - \tilde x)^\top P (x - \tilde x)}_{\V} \leq \lambda_{\max}(P)\vert x - \tilde x \vert ^2,
	\end{align*}
	choosing $\underline{\alpha} = \lambda_{\min}(P)$ and $\overline{\alpha} = \lambda_{\max}(P)$ simply satisfies condition~\eqref{eq:dISS-con1-dis}.
	
   Now we demonstrate that condition~\eqref{eq:ISS-con2-dis} holds as well. Since $\aleph(x)P^{-1} = \mathbf J_{0,T}\mathds Y(x)$ as per condition~\eqref{eq:con1 SOS-dis}, we have $\aleph(x) = \mathbf J_{0,T} \mathds Y(x) P$. We also have $\aleph(x) = \mathbf J_{0,T} G(x)$ as in condition~\eqref{eq:lemma con1-dis}. Hence, one can opt for $G(x) = \mathds Y(x) P$ as an appropriate choice. To show condition~\eqref{eq:ISS-con2-dis}, we have
	\begin{align*}
		\Vp =&~\partial_x\V(A\VecF(x) + B u) \\
		&+ \partial_{\tilde x}\V(A\VecF(\tilde x)+ B \tilde u)\\
		=&~2(x - \tilde x)^\top P (A\VecF(x) + B u)\\
		& - 2(x - \tilde x)^\top P (A\VecF(\tilde x)+ B \tilde u).
	\end{align*}
	By applying $u= \mathbf K(x)x +\hat u$ and $\tilde u = \mathbf K(\tilde x)\tilde x + \hat{\tilde u}$, and considering $\VecF(x) = \aleph(x)x$ and $\VecF(\tilde x) = \aleph(\tilde x)\tilde x$ based on condition~\eqref{new}, one can obtain
	\begin{align*}
		\Vp =&~2(x - \tilde x)^\top P (A\aleph(x) + B \mathbf K(x))x  \\
		&- 2(x - \tilde x)^\top P (A\aleph(\tilde x) + B \mathbf K(\tilde x))\tilde x \\
		&+ 2(x - \tilde x)^\top P B (\hat u - \hat{\tilde u}).
	\end{align*}
	By leveraging the data-based representation of the closed-loop system in Lemma~\ref{Lemma1-dis}, and since $G(x) = \mathds Y(x)P$ and $G(\tilde x) = \mathds Y(\tilde x) P$, we attain
	\begin{align*}
		\Vp =&~2(x - \tilde x)^\top P \mathbf X_{1,T}\mathds Y(x) P x \\
		&- 2(x - \tilde x)^\top P \tilde{\mathbf X}_{1,T} \mathds Y(\tilde x) P \tilde x \\
		&+ 2(x - \tilde x)^\top P  B (\hat u - \hat{\tilde u}).
	\end{align*}
	Considering conditions~\eqref{eq:Y on cl-dis} and \eqref{eq:Y' on cl-dis}, we have
	\begin{align*}
		\Vp =&~2(x - \tilde x)^\top P \Sigma P(x - \tilde x) \\
		&+ 2\underbrace{(x - \tilde x)^\top P}_{a} \underbrace{ B (\hat u - \hat{\tilde u})}_{b}.
	\end{align*}
	Using Cauchy-Schwarz inequality~\cite{bhatia1995cauchy} as $a b \leq \vert a \vert \vert b\vert,$ for any $a^\top,b\in \mathbb R^n$, and by leveraging Young's inequality~\cite{young1912classes} as $\vert a\vert \vert b\vert\leq \frac{\vartheta}{2}\vert a\vert^2+\frac{1}{2\vartheta}\vert b\vert^2,$ for any $\vartheta>0$, we get
	\begin{align*}
		\Vp \leq&~2(x - \tilde x)^\top P \Sigma P(x - \tilde x) \\
		&+ \vartheta(x - \tilde x)^\top P P (x - \tilde x) + \frac{\Vert B \Vert^2}{\vartheta}\vert \hat u - \hat{\tilde u}\vert^2.
	\end{align*}
	By applying expansion and factorization on the above expression, one can attain
	\begin{align*}
		\Vp \leq &~(x - \tilde x)^\top P \big[\Sigma + \Sigma^\top + \vartheta\mathds I_n\big] P(x - \tilde x) \\
		&+ \frac{\Vert B \Vert^2}{\vartheta}\vert \hat u - \hat{\tilde u}\vert^2.
	\end{align*}
	Now, according to condition~\eqref{eq:con2 SOS-dis}, we have
	\begin{align*}
		\Vp &\leq -\epsilon(x - \tilde x)^\top P\underbrace{P^{-1}P}_{\mathds I_n} (x - \tilde x) + \frac{\Vert B \Vert^2}{\vartheta}\vert \hat u - \hat{\tilde u}\vert^2\\
		& = -\epsilon \V + \rho\vert \hat u - \hat{\tilde u}\vert^2,
	\end{align*}
	satisfying condition~\eqref{eq:ISS-con2-dis} with $\rho=\frac{\Vert B \Vert^2}{\vartheta}$.  Then, one can deduce that $\V = (x-\tilde x)^\top P (x-\tilde x)$ is a $\delta$-ISS Lyapunov function and $u =\mathbf U_{0,T}G(x)x + \hat u =\mathbf U_{0,T}\mathds Y(x)Px + \hat u$ is its corresponding $\delta$-ISS controller for the \textsc{ctia-NSP} $\Omega$, which completes the proof.
\end{proof}
\begin{remark}\label{SOS}
	Existing software tools like \textsf{SOSTOOLS} \cite{papachristodoulou2013sostools}, in conjunction with a semidefinite programming (SDP) solver such as \textsf{SeDuMi} \cite{sturm1999using}, can be employed to enforce the conditions in~\eqref{eq:main theorem conditions}.
\end{remark}

\begin{remark}
    Note that $\rho = \frac{\Vert B \Vert^2}{\vartheta}$ is not used in the analysis or the construction of the $\delta$-ISS controller: it serves only as an \emph{existential} value; hence, knowledge of $B$ is not required. However, if the proposed approach is used for purposes requiring the value of $\rho$, it suffices to know an upper bound on $\Vert B \Vert$.
\end{remark}
\begin{algorithm}[t!]
		\caption{Data-driven design of $\delta$-ISS Lyapunov function and its $\delta$-ISS controller}
		\begin{center}
			\begin{algorithmic}[1]\label{Alg1}
				\REQUIRE
				The maximum degree of $\VecF(x)$ or an exaggerated dictionary
				\FOR{time interval $[t_0,t_0 +(T-1)\tau]$ with $T\in\N^+$}
				\STATE Gather $\mathbf U_{0,T},\mathbf X_{0,T},\tilde{\mathbf X}_{0,T}, \mathbf X_{1,T},\tilde{\mathbf X}_{1,T},$ as in~\eqref{eq:data}
				\STATE Compute $\mathbf J_{0,T}, \tilde{\mathbf J}_{0,T}$ based on~\eqref{eq:monomial traj-dis}
				\ENDFOR
				\STATE Employ \textsf{SOSTOOLS} and design $P^{-1}=\Theta\footnotemark,\mathds Y(x)$, and $\mathds Y(\tilde x)$, satisfying conditions in~\eqref{eq:main theorem conditions}, with fixed parameters $\epsilon,\vartheta\in\R^+$
				\STATE Design $\delta$-ISS Lyapunov function $\V = (x-\tilde x)^\top \Theta^{-1}(x-\tilde x)=(x-\tilde x)^\top P(x-\tilde x)$ and its corresponding $\delta$-ISS  controller $u = \mathbf U_{0,T}\mathds Y(x) Px + \hat u$
				\ENSURE $\delta$-ISS Lyapunov function and its $\delta$-ISS controller
			\end{algorithmic}
		\end{center}
	\end{algorithm}

We outline the necessary steps of our data-driven approach for the design of the $\delta$-ISS Lyapunov function and its $\delta$-ISS controller in Algorithm~\ref{Alg1}. 

\section{Simulation Results}\label{sec:4}
\footnotetext{To satisfy conditions~\eqref{eq:con1 SOS-dis}, \eqref{eq:con1 SOS'-dis}, and \eqref{eq:con2 SOS-dis}, we introduce $\Theta = P^{-1}$ and require it to be a \emph{symmetric positive-definite} matrix. Once conditions~\eqref{eq:con1 SOS-dis}, \eqref{eq:con1 SOS'-dis}, and \eqref{eq:con2 SOS-dis} are met and $\Theta$ is designed, the matrix $P$ can then be determined using inversion on $\Theta$, \ie, $\Theta^{-1} = (P^{-1})^{-1} = P$.}
\begin{figure}[t!]
	\centering
	\subfloat[Evolution of $x_1(t)$ and $\tilde x_1(t)$]{
		\includegraphics[width=0.75\linewidth]{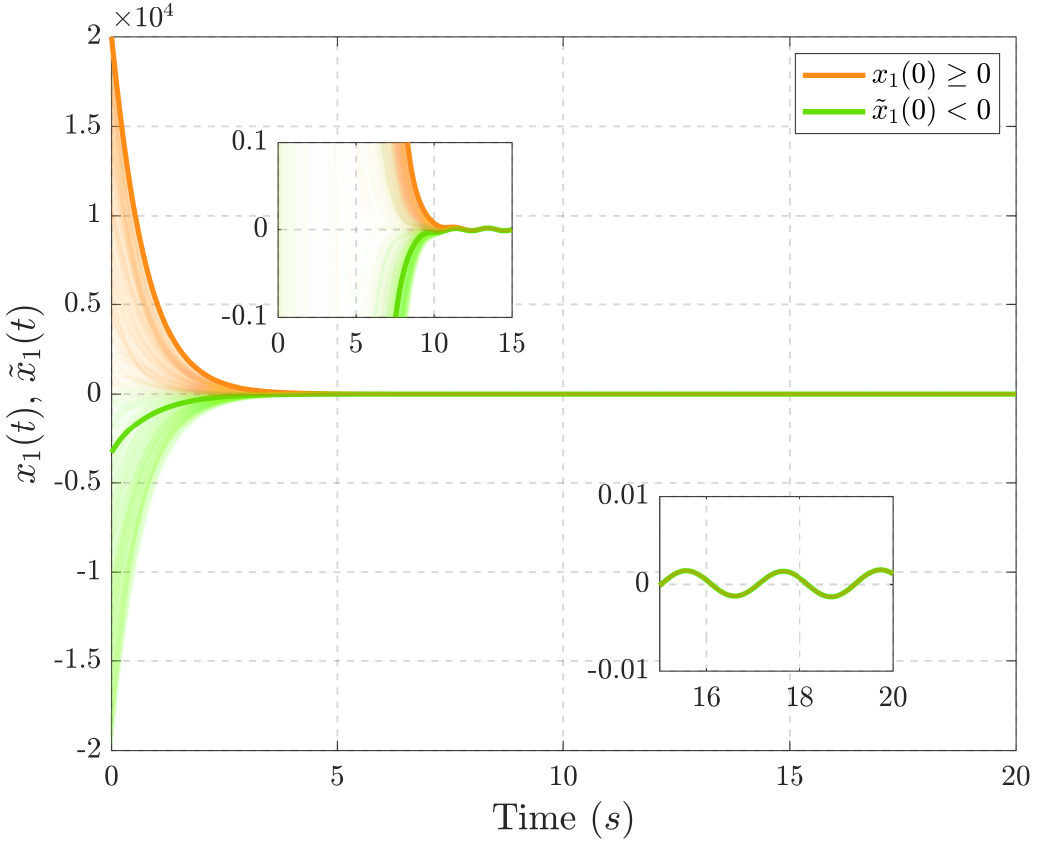}
		\label{fig:x1}}
	\\
	\subfloat[Evolution of $x_2(t)$ and $\tilde x_2(t)$]{
		\includegraphics[width=0.75\linewidth]{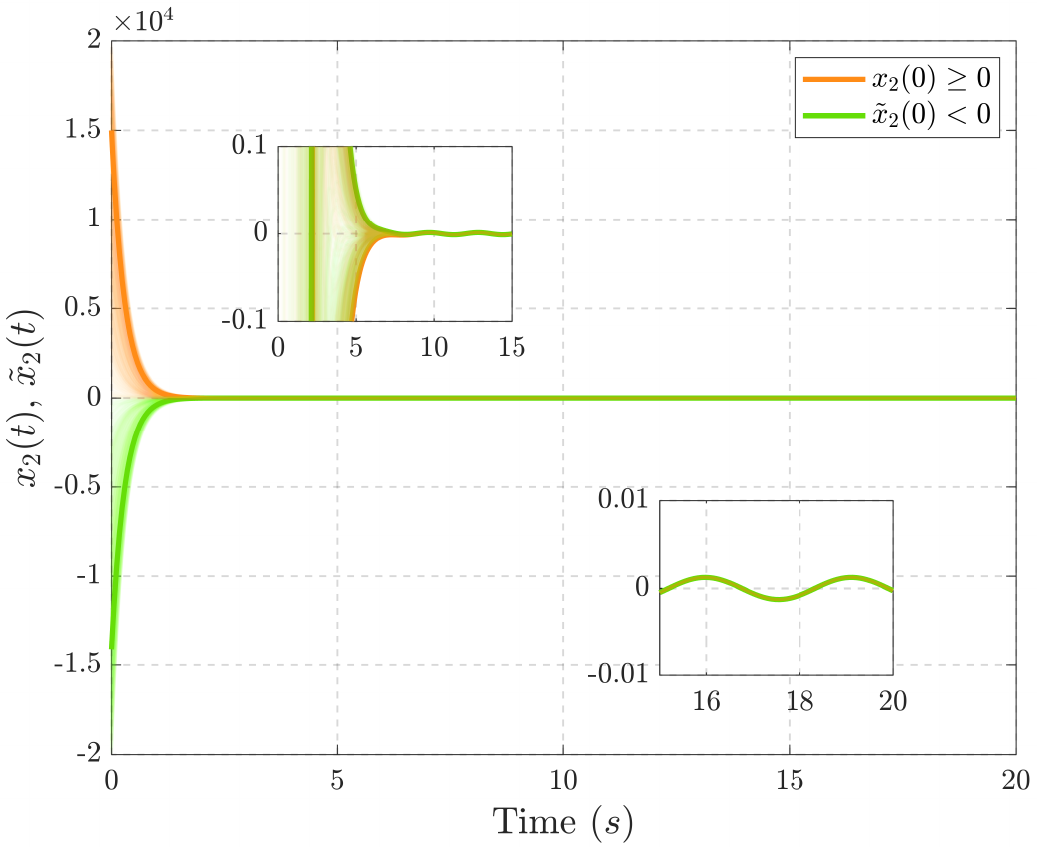}
		\label{fig:x2}}\\
	\subfloat[Evolution of $x_3(t)$ and $\tilde x_3(t)$]{
		\includegraphics[width=0.75\linewidth]{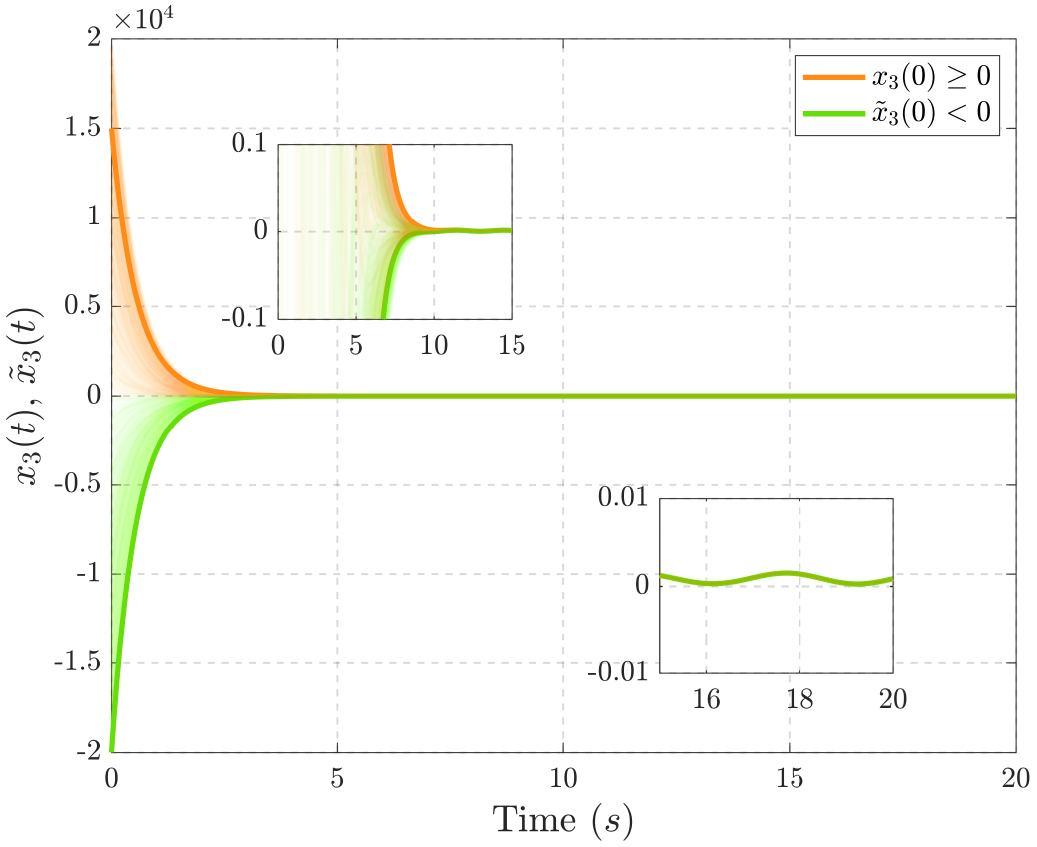}
		\label{fig:x3}}
	\caption{Trajectories $x(t)$ and $\tilde x(t)$ starting from $1000$ arbitrary initial conditions under the designed $\delta$-ISS controller $u(t)$ with $\hat u(t) = \hat{\tilde u}(t) = [\sin(3t)~~\cos(2t)~~\sin^2(t)]^\top$.}
	\label{fig:trajectories}
\end{figure}
We demonstrate the effectiveness of our proposed approach by applying it to a physical system including a rotating rigid spacecraft~\cite{khalil2002control}, characterized by three states, $x = [x_1 ~~x_2 ~~x_3]^\top$, which represent the angular velocities $\omega_1$ to $\omega_3$ along the principal axes. The dynamics of the system are described as follows:
\begin{equation}\label{eq:sc}
	\begin{split}
		\dot x_{1} &= \frac{(J_{2} - J_{3})}{J_{1}}x_{2}x_{3} + \frac{1}{J_{1}}u_{1},\\
		\dot x_{2} &= \frac{(J_{3} - J_{1})}{J_{2}}x_{1}x_{3} + \frac{1}{J_{2}}u_{2},\\
		\dot x_{3} &= \frac{(J_{1} - J_{2})}{J_{3}}x_{1}x_{2} + \frac{1}{J_{3}}u_{3},
	\end{split}
\end{equation}
where $u = [u_1 ~~u_2 ~~u_3]^\top$ represents the torque input, and $J_1 = 200, J_2 = 200$, and $J_3 = 300$ denote the principal moments of inertia. In accordance with the system description~\eqref{eq:sys}, the unknown matrices $A$ and $B$ are
\begin{align*}
    A = \begin{bmatrix}
        \frac{(J_{2} - J_{3})}{J_{1}} & 0 & 0\\
        0 & \frac{(J_{3} - J_{1})}{J_{2}} & 0\\
        0 & 0 & \frac{(J_{1} - J_{2})}{J_{3}}
    \end{bmatrix}\!\!,\, B = \begin{bmatrix}
        \frac{1}{J_{1}} & 0 & 0\\
        0 & \frac{1}{J_{2}} & 0\\
        0 & 0 & \frac{1}{J_{3}}
    \end{bmatrix}\!\!,
\end{align*}
where the actual dictionary is $\VecF(x) = [x_{2}x_{3}~~x_{1}x_{3}~~x_{1}x_{2}]^\top\!,$ all of which are assumed to be unknown. With a slight abuse of notation, given the exaggerated dictionary $\VecF(x) = [x_{1}~~x_{2}~~x_{3}~~x_{1}^2~~x_{1}x_{2}~~x_{1}x_{3}~~x_{2}x_{3}]^\top\!,$ which includes the actual nonlinear terms in the dynamics along with irrelevant terms, we choose $\aleph(x) \!= \![\aleph_1^\top~~\aleph_2^\top~~\aleph_3^\top]^\top\!,$ with $\aleph_1 \!= \!\mathds I_3, \aleph_2 \!=\! \begin{bmatrix}
x_1 & 0 & 0
\end{bmatrix}\!\!,\aleph_3 \!=\! \begin{bmatrix}
x_2 & 0 & 0\\
0 & x_3 & 0\\
0 & 0 & x_1
\end{bmatrix}\!\!.\!$ Following the steps in Algorithm~\ref{Alg1}, we collect $T = 300$ samples with a sampling time $\tau = 0.1$, while setting $\epsilon = 0.9$, $\vartheta = 0.44$, and design
\begin{align*}
	P =& \begin{bmatrix}
		1.9087  & -0.1404 &  -0.1441\\
		-0.1404  &  5.3907  &  0.1229\\
		-0.1441  &  0.1229  &  2.8604
	\end{bmatrix}\!\!,\\
	\Sigma =&  \begin{bmatrix}
		-0.7926  &  0.0245  &  0.0058\\
		-0.0459 &  -0.6426  &  0.0088\\
		-0.0195  &  0.0130 &  -0.6633
	\end{bmatrix}\!\!,
\end{align*}
\begin{align*}
	u_1  =&-7.7964\times 10^{-11}x_1^2 + 7.9309\times 10^{-11}x_1x_2 \\
	&+ 4.0449\times 10^{-10}x_1x_3 - 7.1981\times 10^{-12}x_2^2 \\
	&+ 100x_2x_3 - 1.699\times 10^{-10}x_3^2 - 303.4255x_1 \\
	&+ 48.8131x_2 + 26.7615x_3+ \hat u_1\\
   u_2 = &-1.648\times 10^{-11}x_1^2 - 1.2567\times 10^{-10}x_1x_2 \\
   &- 100x_1x_3 - 3.4092\times 10^{-11}x_2^2 + 7.6116\times 10^{-9}x_2x_3 \\
   &+ 3.4715\times 10^{-10}x_3^2 + 0.2917x_1 - 691.3298x_2 \\
   &- 9.4448x_3 + \hat u_2,\\
	u_3 = & ~9.6134\times 10^{-12}x_1^2 - 6.2879\times 10^{-12}x_1x_2 \\
	&- 3.3625\times 10^{-10}x_1x_3 - 3.2756\times 10^{-12}x_2^2 \\
	&+ 8.3844\times 10^{-9}x_2x_3 + 2.3587\times 10^{-10}x_3^2 \\
	&+ 16.9632x_1 - 2.5509x_2 - 567.8759x_3 + \hat u_3.
\end{align*}
We choose $\hat u(t) = \hat{\tilde u}(t) = [\sin(3t)~~ \cos(2t)~~\sin^2(t)]^\top$, and apply the controller to the system for $1000$ arbitrary initial conditions $x(0)\in[0, 2\times 10^4]$ and $\tilde x(0)\in[-2\times 10^4, 0)$. Since $\hat u(t) = \hat{\tilde u}(t)$, the trajectories are expected to demonstrate the $\delta$-GAS property, as per Definition~\ref{def:delta-ISS}. As can be seen in Fig.~\ref{fig:trajectories}, all trajectories converge to one another (demonstrating $\delta$-ISS property while ensuring robustness against $\hat u, \hat{\tilde u}(t)$) and remain around the equilibrium point, \ie, the origin (demonstrating the $\delta$-GAS property). For better illustration, the norm of the difference between every two trajectories has been sketched in logarithmic scale in Fig.~\ref{fig:error}, which is monotonically decreasing and elucidates that all trajectories converge to each other asymptotically.

\begin{figure}
	\centering
	\includegraphics[width=0.7\linewidth]{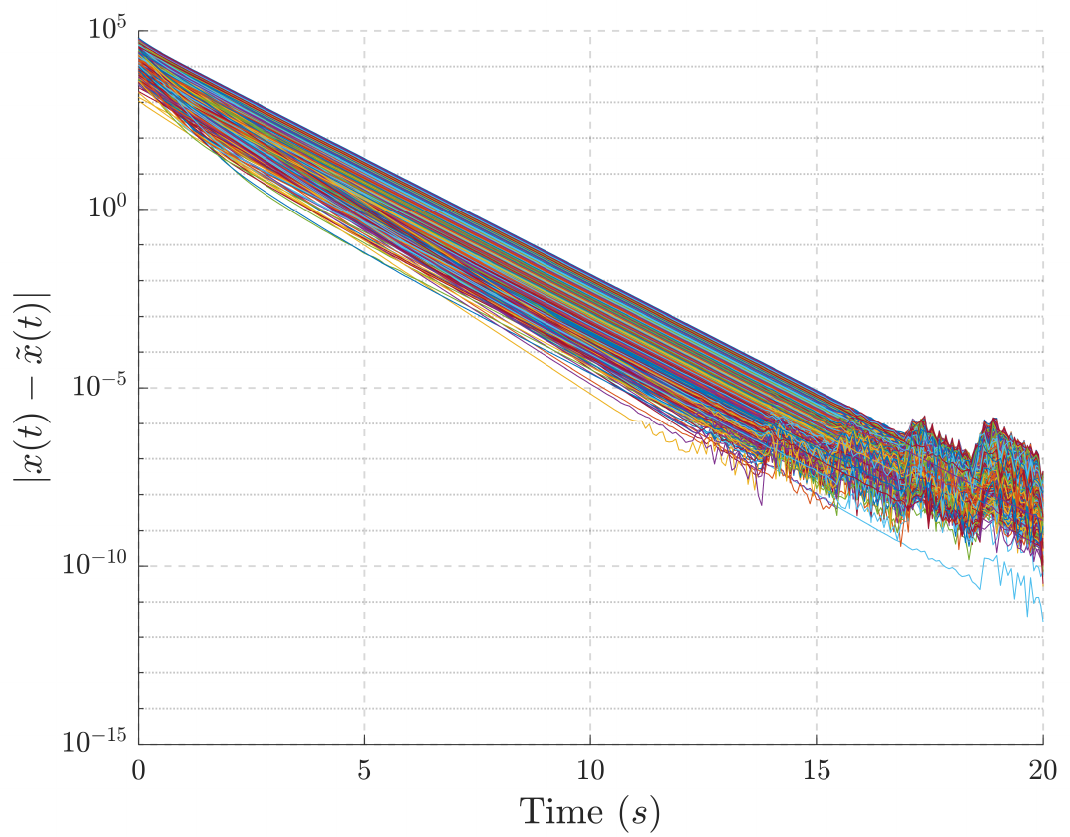}
	\caption{Norm of the difference between trajectories with respect to each other in logarithmic scale.}
	\label{fig:error}
\end{figure}

\section{Conclusion}\label{sec:5}
In this work, we developed a data-driven methodology for designing $\delta$-ISS Lyapunov functions together with $\delta$-ISS controllers for continuous-time input-affine nonlinear systems with polynomial dynamics. Our approach ensured the $\delta$-ISS property without requiring explicit system dynamics, relying solely on two input-state trajectories from sufficiently excited system behavior. Using the collected samples and satisfying a particular rank condition, we designed $\delta$-ISS controllers by formulating a sum-of-squares optimization program. A physical case study with unknown dynamics validated our methodology, demonstrating its robustness subject to external inputs and its effectiveness in ensuring incremental stability of nonlinear systems. Expanding our framework to design $\delta$-ISS controllers for nonlinear systems beyond polynomials and considering measurement noises are being explored as a future research direction.

\bibliographystyle{ieeetr}
\bibliography{biblio}

\end{document}